\documentclass[lettersize,journal]{IEEEtran}
\usepackage{amsmath,amsfonts}
\usepackage{algorithmic}
\usepackage{algorithm}
\usepackage{array}
\usepackage[caption=false,font=normalsize,labelfont=sf,textfont=sf]{subfig}
\usepackage{textcomp}
\usepackage{stfloats}
\usepackage{url}
\usepackage{verbatim}
\usepackage{graphicx}
\usepackage{cite}
\usepackage{hyperref}
\usepackage{amsthm}
\usepackage{fancyhdr}
\usepackage{enumerate}
\usepackage{tabularx}
 \usepackage{booktabs}
 \usepackage{balance}
\usepackage{float} 

\newtheorem{theorem}{Theorem}

\newcommand{\secref}[1]{Section~\ref{#1}}
\newcommand{\figref}[1]{Fig.~\ref{#1}}
\newcommand{\tabref}[1]{Table~\ref{#1}}
\newcommand{\thrmref}[1]{Theorem~\ref{#1}}

\newcommand{\specialcell}[2][c]{%
  \begin{tabular}[#1]{@{}c@{}}#2\end{tabular}}

\begin{document}
\title{Distributed Experimental Design: Bayes-optimal Fusion of Local Designs}
\author{Nagananda K G, Lav R. Varshney,~\IEEEmembership{Senior Member,~IEEE,} Pramod K. Varshney,~\IEEEmembership{Life Fellow,~IEEE}  
\thanks{Nagananda K G is with Fariborz Maseeh Department of Mathematics and Statistics, Portland State University, Portland, OR 97201, USA. (email: \texttt{nanda@pdx.edu}).}
\thanks{Lav R. Varshney is with the AI Innovation Institute, Stony Brook University, Stony Brook, NY 11794, USA, and with Brookhaven National Laboratory, Upton, NY 11973, USA. (email: \texttt{lav.varshney@stonybrook.edu}).}
\thanks{Pramod K. Varshney is with the Department of Electrical Engineering and Computer Science, Syracuse University, Syracuse, NY 13244 USA. (email: \texttt{varshney@syr.edu}).}
}



\maketitle

\begin{abstract}
We develop a decision-theoretic framework for distributed Bayesian experimental design in which local agents evaluate candidate experiments using expected information gain and transmit their local design decisions to a fusion center. Unlike centralized Bayesian design, where all likelihood components and information-gain values are available to a single planner, the fusion center in the distributed setting chooses a global experiment from compressed local recommendations.  We derive the Bayes-optimal fusion rule, which selects the experiment with largest conditional expected centralized information gain given the observed local design decisions.   This rule is analogous in spirit to optimal fusion rules in distributed detection, but differs fundamentally because the underlying utility is expected information gain and the resulting loss is information-gain regret rather than classification error.  We also establish information-loss bounds and identify conditions under which the decision-only fusion rule is asymptotically equivalent to the centralized design.  Numerical experiments show that Bayes-optimal fusion closely approximates the centralized oracle, whereas majority voting can be highly suboptimal when a minority of sites carry disproportionate information.  
\end{abstract}

\section{Introduction}
Optimal experimental design is a critical catalyst for accelerating scientific discovery, as it ensures that every trial yields the maximum possible information while optimizing resource usage \cite{Musslick2025}. This strategic and systematic approach transforms experimentation from a trial-and-error process to a well-planned approach leading to  significant scientific progress and breakthroughs. Many modern scientific experimental investigations rely on data collected across multiple facilities, instruments, sensors, or computational agents. In such settings, experimental design is no longer a centralized problem as a central planner may not have direct access to all local data streams, local likelihood models, computational resources, or operational restrictions; see, for instance, \cite{Alexander2023, Lobel2025, Rillo2026}. This gives rise to a distributed experimental design problem where local agents must evaluate experimental alternatives using local information, while a fusion center must combine these local design recommendations into a global experimental design. The local recommendation that is best for one site, however, need not be best for the system as a whole. This gap between local and global scientific value is the main issue addressed in this paper and an optimal distributed experimental design solution is sought.

Distributed design of experiments is important for several reasons.  First, many contemporary experiments are physically distributed.  For example, multi-center biomedical studies, environmental monitoring systems, distributed quality-control systems, and decentralized scientific instruments generate information at separate locations \cite{Hart2006, Rundel2009, Deelman2009, Ram2012, Meinert2012, Piantadosi2017, Montgomery2020, Li2023, Roy2026}.  Second, communication constraints may prevent the full transmission of local likelihoods, data summaries, or posterior distributions \cite{McMahan2017, Jordan2019}.  A local agent may be able to send only a small message, such as choice of a preferred experiment or a low-dimensional statistic.  Third, privacy, proprietary restrictions, or regulatory constraints may prevent local sites from sharing raw data or detailed model information \cite{Agarwal2018, Chen2024}.  Fourth, in large-scale systems, centralized design may be computationally prohibitive, whereas local design can be performed in parallel \cite{Boyd2011, Le2021}.  These considerations motivate distributed procedures in which local experimental design computations are performed, and then the fusion center combines compressed design messages to select a common global experiment.

Bayesian experimental design is especially natural in this setting.  In Bayesian design, an experiment is chosen by comparing its expected utility under the prior or current posterior distribution \cite{Ryan2016}.  A common and principled utility is expected information gain (EIG)\textemdash measuring the expected reduction in uncertainty about an unknown parameter $\theta$ after the experiment is performed \cite{Chaloner1995}\textemdash which we consider.  If $\xi$ denotes a controllable experimental design, then Bayesian experimental design seeks a choice of $\xi$ that maximizes EIG about $\theta$.  The EIG criterion is directly tied to posterior learning rather than to a particular point estimator, hypothesis test, or loss function after data collection.

The distributed Bayesian experimental design problem differs fundamentally from the classical centralized formulation, where a designer has access to all relevant model components and chooses the experiment that maximizes centralized EIG \cite{Kiefer1959, Fedorov1972}.  In a distributed formulation, however, local agents may only compute local EIGs and transmit local design decisions.  The fusion center must then choose a global experiment from these local recommendations.  After the fusion center selects the global design, that design is communicated back to the sites, who implement the selected experiment rather than their purely local choices.  The key question is therefore not merely how to compute EIG, but how to fuse local design decisions so that the resulting global experiment is close to the centralized EIG optimum.

The main novelty of the paper is the fusion of local design decisions under an information-gain loss. This is distinct from simply aggregating local estimates, local posterior distributions, or local test decisions. The information transmitted by each local agent is a design recommendation.  The design choice made by the fusion center is a global experiment. The utility function is the loss in EIG caused by not selecting the optimum global experiment. Thus the problem lies at the intersection of Bayesian experimental design, distributed decision making, and information-theoretic utility maximization (see, for example, \cite{Williams2007})\textemdash but is not reducible to any of these standard formulations.

The main contributions of this paper are as follows. 
\begin{enumerate}[(1)]
\item We introduce distributed Bayesian experimental design as a decision-theoretic fusion problem.  At each local agent, the local EIG is used to select a local experiment. The fusion center receives the resulting local design decisions and must choose a global experiment. We define the centralized EIG, the local EIGs, the local design messages, and the fusion rule. This formulation makes explicit the information loss induced by communication constraints.  It separates sources of sub-optimality: the loss due to using local rather than centralized information and the loss due to compressing local information into design decisions.

\item In \thrmref{thrm:Bayes_optimal}, we characterize the Bayes-optimal fusion rule for the minimization of information-gain loss. The optimal fusion rule chooses the experiment with the largest posterior conditional expected centralized information gain given the local design messages. In other words, if the fusion center observes a vector of local design decisions, then the Bayes-optimal rule compares the conditional mean centralized information gains of the competing experiments given that message vector. This result is analogous in spirit to optimal fusion rules in distributed decision theory \cite{Varshney1996}, but the quantity being fused and the loss being minimized are different. 

\item In \thrmref{thrm:loss_asymptotic}, we establish information-loss bounds and asymptotic equivalence results. The information-gain regret of a fusion rule is defined as the difference between the centralized oracle EIG and the EIG attained by the fused design. This regret admits a natural decomposition into a conditional Bayes risk term and an approximation term. Under suitable separation and concentration conditions, the Bayes fusion rule based only on local design decisions becomes asymptotically equivalent to the centralized design. This result formalizes the intuition that, when the local design messages contain enough information about the global ordering of experiments, distributed decision fusion can recover the centralized design with vanishing information loss.
\end{enumerate}

We assume that all agents choose from a common finite experiment set, since the fusion center selects one system-wide experiment to be implemented by every agent.  The requirement of common treatment protocol is standard in multicenter clinical trials; see, for example, \cite{Meinert2014}.  Allowing different experiment sets or site-specific choices across agents would correspond to a different assignment problem with decision space $\{1, \ldots, \xi\}^M$ for $M$ local agents., whereas our focus is on fusing compressed local recommendations into one global design.  This restriction is appropriate when the experimental condition is a common sensing protocol, sampling frequency, treatment regime, stimulus, platform setting, or operating mode that must be deployed uniformly.

A useful way to view the theory is through the distinction between majority agreement and information agreement.  Majority agreement means that a majority of local agents prefer one experiment, whereas information agreement means that a smaller but more informative group favor the experiment with larger centralized EIG.  The two agreements need not coincide. Consequently, majority voting over local design recommendations can be suboptimal.  The Bayes-optimal fusion rule accounts for this by relating local decision patterns to centralized EIG.

This distinction is central to the paper. In many distributed systems, the most common local recommendation is not necessarily the globally most informative design. For instance, in heterogeneous sensor networks \cite{Krause2008, Garnett2010}, many low-precision sensors may favor one experiment, while a few high-precision sensors favor another. A majority-vote rule would select the first experiment, even if the second experiment produces greater centralized information gain. The proposed Bayesian fusion rule instead selects the experiment that maximizes conditional expected centralized information gain given the observed pattern of local recommendations. Thus the rule is sensitive not only to how many agents favor an experiment, but also to how much the pattern contributes to global information.

The present work is related to, but distinct from, the classical theory of distributed detection \cite{Tenney1981, Tsitsiklis1985, Reibman1987, Tsitsiklis1993, Viswanathan1997, Blum1997, Veeravalli2012}.  In distributed detection, local sensors observe data and transmit local detection decisions to a fusion center, whose task is to make a final hypothesis decision; the classical Chair--Varshney design gives an optimal fusion rule under conditional independence and known local sensor characteristics \cite{Chair1986}.  In our setting, the local messages are not decisions about the state of nature, but recommendations about which experiment to perform, and the fusion center selects a global design to maximize expected learning about an unknown parameter.  Thus, although both settings involve local messages and a fusion center, the optimality criteria are different: Chair--Varshney's fusion rule is driven by likelihood ratios for hypothesis testing, whereas our fusion rule is driven by conditional expected centralized EIG given the local design messages.  Consequently, the relevant loss is information-gain regret rather than classification error, and the resulting fusion problem is not a direct instance of distributed detection.  Distributed inference has also been addressed from an information-theoretic viewpoint; see, for example, \cite{Hoballah1989, Kreucher2007}.  While both these works and the present paper draw on information-theoretic concepts, they address fundamentally different problems.

We illustrate the theory using two examples. The first is a scalar Gaussian model with two candidate experiments, where $\theta \sim \mathcal{N}(0, \tau^2)$ and each local experiment has an experiment-specific precision contribution $r_{m,e}$.  In this case, local and centralized EIGs are available in closed form as $0.5\log(1+\tau^2 r_{m,e})$ and $0.5\log\{1+\tau^2\sum_{m=1}^M r_{m,e}\}$, respectively.  The second is a distributed binary-response model, where $Y_{m,e}\mid\theta$ is Bernoulli with logistic success probability. For this model, the design utility is computed from the prior-averaged Bernoulli Fisher information.  Together, these examples show that the proposed fusion principle applies both to Gaussian and non-Gaussian response settings, and that the experiment preferred by the majority of local agents need not be the experiment with the largest centralized EIG.

The computer experiments compare four rules: the centralized oracle, the Bayes fusion rule based only on local design decisions, majority voting, and random fusion.  The centralized oracle observes all local precision contributions and selects the experiment with the largest centralized EIG.  The Bayes fusion rule observes only the number or pattern of local design recommendations and chooses the experiment with the largest estimated conditional centralized EIG.  Majority voting chooses the experiment recommended by most local agents.  Random fusion ignores the local design recommendations and selects one of the candidate experiments at random with equal probability, serving as a baseline.

The numerical results show that Bayes fusion nearly matches the centralized oracle using only binary local recommendations, whereas majority voting can incur substantial information-gain regret under heterogeneous local informativeness. As $M$ increases, Bayes-fusion regret decreases, consistent with the asymptotic equivalence theory; majority voting need not improve when the numerical majority is associated with less informative sites.  These findings have two implications. (i) From a methodological standpoint, distributed Bayesian experimental design should be formulated as an information-gain fusion problem, not as a voting problem. The local recommendations are useful, but they must be interpreted through the information-gain criterion. (ii)  From a practical viewpoint, communication-efficient design is possible. A fusion center need not always receive all local likelihoods, posterior distributions, or numerical EIGs.  Under sufficient separation between the centralized EIGs of the competing experiments and sufficiently informative local design messages so that the global EIG ordering can be recovered with high probability, Bayes fusion of local design decisions achieves near-centralized performance.  

\secref{sec:system_model} formulates the distributed Bayesian design problem.   \secref{sec:main_results} derives the Bayes-optimal fusion rule under information-gain loss.  \secref{subsec:asymptotic_equivalence} establishes the information-loss bounds and asymptotic equivalence to the centralized design.  \secref{sec:illustrative_models} presents illustrative examples to reveal analytical structure underlying the theory.  \secref{sec:numerical_studies} presents numerical results.  \secref{sec:conclusion} concludes the paper. 

\section{System Model}\label{sec:system_model}
Let $\theta \in \Theta$ denote the unknown quantity of interest.  The prior density of $\theta$ is denoted $p(\theta)$.  Let $\mathcal{E} = \{1, \ldots, L\}$ be a finite set of candidate experiments.  There are $M$ distributed experimental sites.  At site $m = 1, \ldots, M$, the local EIG associated with experiment $\xi \in \mathcal{E}$ is
\begin{equation*}
I_m(\xi) = \mathbb{E}_{Y_{m, \xi}} \left[\mathrm{KL} \{p(\theta \mid Y_{m, \xi}, \xi)\,\|\, p(\theta)\}\right],
\end{equation*}
where $I_m(\xi)$ is the mutual information between $\theta$ and the local observation $Y_{m, \xi}$ induced by experiment $\xi$, $\mathrm{KL}$ denotes the Kullback--Leibler divergence, and the expectation is taken with respect to the prior predictive distribution of the local observation under experiment $\xi$.  The local design decision at site $m$ is then
\begin{equation*}
D_m = \arg\max_{\xi \in\mathcal{E}} I_m(\xi).
\end{equation*}
If all local model information were available centrally, the centralized EIG for experiment $\xi$ is $I_{\mathrm{cen}}(\xi) = \mathbb{E}_{Y_{\xi}}\left[\mathrm{KL}\{p(\theta \mid Y_{\xi}, \xi)\,\|\, p(\theta)\}\right]$, where $Y_{\xi}=(Y_{1, \xi}, \ldots, Y_{M, \xi})$ denotes the collection of all observations that would be obtained under experiment $\xi$. The centralized oracle chooses
\begin{equation*}
B = \arg\max_{\xi \in\mathcal{E}} I_{\mathrm{cen}}(\xi).
\end{equation*}
The fusion center, however, does not observe the full local likelihoods or the numerical values of the local EIGs. It observes only the local design messages $D = (D_1, \ldots, D_M)$, and must choose a global experiment through a fusion rule $\delta(D) \in \mathcal{E}$.  A natural utility function for the fusion rule is the loss in centralized EIG relative to the oracle:  $L\{\delta(D)\} = I_{\mathrm{cen}}(B)-I_{\mathrm{cen}}\{\delta(D)\}$.  Thus, the Bayes-optimal fusion rule is
\begin{equation*}
\delta^\ast(d) = \arg\max_{\xi \in \mathcal{E}} \mathbb{E}\{I_{\mathrm{cen}}(\xi) \mid D=d\}.
\end{equation*}
Note that, the $L$ experiments case formulated here captures the central challenge of distributed design fusion: the experiment preferred by the majority of local agents need not be the experiment with the largest centralized EIG.  

In the remainder of the paper, for brevity as well as for ease of elucidation we consider only the binary case $\mathcal{E} = \{0, 1\}$, using which we develop the optimal fusion rule, information-loss bounds, and asymptotic equivalence results.  At site $m$, one of two possible experiments may be conducted. The local design variable is $\xi_m \in \{0,1\}$, $m = 1, \ldots, M$.  After experiment $\xi_m$ is selected at site $m$, the site observes a random response $Y_m \in \mathcal{Y}_m$.  The conditional density of $Y_m$, given $\theta$ and $\xi_m$, is $p_m(y_m\mid \theta,\xi_m)$.  We assume conditional independence across sites given $\theta$ and the chosen designs. Thus, for a design vector $\xi=(\xi_1,\ldots,\xi_M) \in \{0, 1\}^M$, the joint conditional density of $Y = (Y_1,\ldots,Y_M)$
is
\begin{equation*}
p(y\mid \theta,\xi) = \prod_{m=1}^M p_m(y_m \mid \theta,\xi_m).
\end{equation*}
The corresponding marginal density of $Y$ under design $\xi$ is
\begin{equation*}
p(y\mid \xi) = \int_{\Theta}p(\theta)\prod_{m=1}^M p_m(y_m\mid \theta,\xi_m) d\theta.
\end{equation*}
The posterior density is
\begin{equation*}
p(\theta\mid y,\xi) = \frac{p(\theta)\prod_{m=1}^M p_m(y_m\mid \theta,\xi_m)}{\int_{\Theta}p(\vartheta)\prod_{m=1}^M p_m(y_m\mid \vartheta,\xi_m) d\vartheta}.
\end{equation*}

At the local site $m$, the EIG for choosing local experiment $\xi_m \in \{0, 1\}$ is $U_m(\xi_m) = I(\theta; Y_m \mid \xi_m) = H(\theta)\mathbb{E}_{Y_m\mid \xi_m}\left[H(\theta\mid Y_m,\xi_m)\right]$, where $I(\cdot; \cdot)$ and $H(\cdot)$ denote mutual information and entropy, respectively.  It can also be written as the expected KL divergence from prior to posterior:  $U_m(\xi_m) = \mathbb{E}_{Y_m\mid \xi_m}\left[D_{\mathrm{KL}}\left(p(\theta \mid Y_m, \xi_m) \Vert p(\theta)\right)\right]$.  This is the expected Bayesian surprise \cite{Varshney2013}.  The locally optimal design at site $m$ is, therefore,
\begin{equation*}
\xi_m^{\mathrm{loc}} = \arg\max_{\xi_m \in \{0, 1\}} U_m(\xi_m).
\end{equation*}
The local design decision is defined by
\begin{equation}
D_m = 
\begin{cases}
1, & U_m(1 ) >U_m(0),\\
0, & U_m(0) > U_m(1).
\end{cases}
\label{eq:D_m}
\end{equation}
For definiteness, ties are resolved by setting $D_m = 1$.  Thus, $D_m$ is a local recommendation about which experiment is more informative at site $m$.  Following \eqref{eq:D_m}, $D_m$ is a random variable through the randomness of the local information, local model characteristics, or site-specific design environment.  The realized value of $D_m$ is denoted $d_m$.  A fusion rule is written as $\delta(d_1, \ldots, d_M)$ when applied to observed messages.  Probabilities and conditional expectations are written in terms of the random vector $D = (D_1, \ldots, D_M)$.  In places where no ambiguity arises, we use $D_m$ and $d_m$ interchangeably to refer to the local design message.

If all models and all design choices are available at the central location, the EIG associated with the full design vector $\xi=(\xi_1,\ldots,\xi_M)$ is $U(\xi) = I(\theta;Y_1,\ldots,Y_M\mid \xi) = H(\theta) \mathbb{E}_{Y\mid \xi} \left[ H(\theta\mid Y,\xi) \right]$.  In KL form, $U(\xi) = \mathbb{E}_{Y\mid \xi} \left[D_{\mathrm{KL}}\left\{p(\theta\mid Y,\xi),\Vert,p(\theta)\right\}\right]$.  The centralized Bayesian design is
\begin{equation*}
\xi^{\mathrm{cent}} = \arg\max_{\xi \in \{0, 1\}^M} U(\xi).
\end{equation*}
This design would be the optimum if the fusion center had access to all local models and could optimize the entire vector of experimental choices jointly.

To obtain the system-wide binary design, we impose the restriction that the same experiment is run at every site. Thus the fusion center chooses a single system-wide design $a \in \{0, 1\}$, and then sets $\xi_1 = \cdots = \xi_M = a$.  Under this restriction, define $U(a) = I(\theta; Y_1, \ldots, Y_M\mid \xi_1 = \cdots = \xi_M = a)$.  The centralized system-wide design is
\begin{equation*}
a^\ast = \arg\max_{a \in \{0, 1\}} U(a)  = 
\begin{cases}
1, & U(1)>U(0),\\
0, & U(0)>U(1).
\end{cases}
\end{equation*}
The important restriction is that the fusion center selects one common system-wide experiment $a^\ast \in \{0, 1\}$, not a site-specific assignment of experiments.  

Each site sends a binary design recommendation $D_m \in \{0, 1\}$, and the fusion center applies a rule $\delta: \{0, 1\}^M \to \{0, 1\}$ to the message vector $D = (D_1, \ldots, D_M)$.  The selected experiment is implemented system-wide, so $\xi_1 = \cdots = \xi_M = \delta(D)$. Since the objective is centralized information gain, the loss from choosing $a$ when the centralized optimum is $b$ is $L(a,b) = U(b) - U(a)$, rather than a classification loss.  Thus the two possible error losses are $U(1) - U(0)$ and $U(0) - U(1)$, depending on which experiment is globally better.

\section{Main results}\label{sec:main_results}\label{sec:main_results}
We first characterize the Bayes-optimal fusion rule, showing how the fusion center selects a global experiment from the observed local design messages. We then study the performance of this rule, where we give an information-loss bound and establish conditions under which the distributed fusion rule is asymptotically equivalent to the centralized EIG design.

\subsection{Bayes-optimal fusion rule}\label{subsec:fusion_rule}
\begin{theorem}
Suppose $M$ distributed sites produce binary local design recommendations $D_1,\ldots,D_M$, where $D_m=1$ indicates that site $m$ recommends experiment $1$ and $D_m=0$ indicates that site $m$ recommends experiment $0$. Let $B \in \{0, 1\}$ denote the centralized information-optimal design. Assume that, conditional on $B$, the local recommendations are independent and satisfy
\begin{equation*}
\begin{split}
P(D_m=1\mid B=1)=p_m, \\
P(D_m=1\mid B=0)=q_m.
\end{split}
\end{equation*}
Let $L(a, b)$ be the information loss incurred by choosing design $a$ when design $b$ is centralized information-optimal. Then the Bayes-optimal fusion rule is $\delta^\ast(d) = 1$ if and only if
\begin{equation}
\begin{split}
\log\frac{\pi_1}{\pi_0}  &+ \sum_{m = 1}^{M} \left[d_m \log\frac{p_m}{q_m} + (1 - d_m)\log\frac{1 - p_m}{1 - q_m}
\right] \\ &> \log\frac{L(1,0)}{L(0,1)}.
\end{split}
\label{eq:Bayes_optimal}
\end{equation}
Otherwise, $\delta^\ast(d) = 0$.  Here $\pi_b = P(B = b)$, $b = 0, 1$.
\label{thrm:Bayes_optimal}
\end{theorem}
\begin{proof}
Define the globally optimal design state
\begin{equation*}
B =
\begin{cases}
1, & U(1)>U(0),\\
0, & U(0)>U(1),
\end{cases}
\end{equation*}
which is the centralized information-optimal design. The fusion center does not observe $B$ directly, but observes $D = (D_1,\ldots,D_M)$.  Let $\pi_b = P(B=b)$, $b \in \{0, 1\}$.  Assume that the local design recommendations are conditionally independent given $B$. That is,
\begin{equation*}
P(D=d \mid B = b) = \prod_{m=1}^M P(D_m=d_m\mid B=b),
\end{equation*}
where $d = (d_1, \ldots, d_M) \in \{0, 1\}^M$.  Write $p_m = P(D_m = 1\mid B = 1)$ and $q_m = P(D_m = 1\mid B = 0)$.  Then, $P(D_m=0\mid B=1)=1-p_m$ and $P(D_m = 0 \mid B = 0) = 1-q_m$.  For a realized vector $d=(d_1,\ldots,d_M)$, the posterior odds of $B = 1$ against $B = 0$ are
\begin{equation*}
\frac{P(B = 1 \mid D = d)}{P(B = 0 \mid D = d)} =  \frac{\pi_1}{\pi_0} \prod_{m = 1}^M \left(\frac{p_m}{q_m}\right)^{d_m} \left(\frac{1 - p_m}{1 - q_m}\right)^{1 - d_m}.
\end{equation*}

The Bayes risk of choosing design $a \in \{0, 1\}$ after observing $D = d$ is
\begin{equation*}
R(a \mid d) = \sum_{b=0}^1 L(a, b)P(B = b \mid D = d).
\end{equation*}
Since choosing the correct design incurs no information loss, $L(0, 0) = L(1, 1) = 0$.  Therefore,
\begin{eqnarray*}
R(1 \mid d) &=& L(1, 0) P(B = 0 \mid D = d), \\
R(0 \mid d) &=& L(0, 1)P(B = 1 \mid D = d).
\end{eqnarray*}
The Bayes-optimal fusion rule chooses $1$ whenever $R(1 \mid d) < R(0 \mid d)$.  That is, $L(1, 0)P(B = 0 \mid D = d) < L(0, 1)P(B = 1 \mid D = d)$.  Equivalently,
\begin{equation*}
\frac{P(B = 1 \mid D = d)}{P(B = 0 \mid D = d)} > \frac{L(1, 0)}{L(0, 1)}.
\end{equation*}
By Bayes' theorem, we have 
\begin{equation*}
\frac{P(B = 1 \mid D = d)}{P(B = 0\mid D = d)} = \frac{\pi_1 P(D = d\mid B = 1)}{\pi_0 P(D = d\mid B = 0)}.
\end{equation*}
Using conditional independence,
\begin{equation*}
\begin{split}
P(D=d\mid B=1) = \prod_{m=1}^M p_m^{d_m}(1-p_m)^{1-d_m}, \\
P(D=d\mid B=0) = \prod_{m=1}^M q_m^{d_m}(1-q_m)^{1-d_m}.
\end{split}
\end{equation*}
Hence, we can write 
\begin{equation*}
\frac{P(B=1\mid D=d)}{P(B=0\mid D=d)} = \frac{\pi_1}{\pi_0} \prod_{m=1}^M\left(\frac{p_m}{q_m}\right)^{d_m}\left(\frac{1-p_m}{1-q_m}\right)^{1-d_m}.
\end{equation*}
Substitution into the Bayes-risk inequality gives
\begin{equation*}
\frac{\pi_1}{\pi_0}\prod_{m=1}^M\left(\frac{p_m}{q_m}\right)^{d_m} \left(\frac{1-p_m}{1-q_m}\right)^{1-d_m} > \frac{L(1,0)}{L(0,1)}.
\end{equation*}
Taking logarithms gives the stated rule \eqref{eq:Bayes_optimal}, thereby completing the proof of \thrmref{thrm:Bayes_optimal}. 
\end{proof}
\thrmref{thrm:Bayes_optimal} has the same structural form as the Chair--Varshney decision-fusion rule \cite{Chair1986}, but the meaning is different.  In decision fusion, the fusion center combines local decisions to decide which hypothesis is true. The loss is usually based on false alarm and miss probabilities.  Here, the fusion center combines local recommendations to decide which experiment to run. The loss is not a detection loss, but the expected information loss caused by running the less informative experiment.  Thus the threshold $\log [L(1,0)/L(0,1)]$ has an experimental-design interpretation. It shifts the fusion rule toward the design whose incorrect rejection would cause the larger loss in EIG.  If $L(1,0)=L(0,1)$, then the threshold reduces to zero apart from the prior odds term, and the rule becomes a weighted likelihood-ratio fusion rule for the local design recommendations.  If $L(1,0)>L(0,1)$, then choosing design $1$ when design $0$ is optimal is more costly than choosing design $0$ when design $1$ is optimal. The fusion center therefore requires stronger evidence before selecting design $1$.  If $L(0,1)>L(1,0)$, then the fusion center is more willing to select design $1$, because failing to select it would cause the larger information loss.

\subsection{Information-loss bound and asymptotic equivalence}\label{subsec:asymptotic_equivalence}
Recall that the centralized information-optimal design is $B \in \{0, 1\}$, where $B = 1 \Longleftrightarrow U(1) > U(0)$ and $B = 0 \Longleftrightarrow U(0) > U(1)$, where $\Longleftrightarrow$ idenotes ``if and only if''.   The fusion center observes local design recommendations $D=(D_1,\ldots,D_M) \in \{0, 1\}^M$ and selects a fused design $\delta(D) \in \{0, 1\}$.  The information loss from choosing design $a$ when design $b$ is centralized information-optimal is $L(a,b) = U(b) - U(a)$.  Thus $L(0, 0) = L(1, 1) = 0$, while $L(0, 1) = U(1) - U(0) > 0$ and $L(1, 0) = U(0) - U(1) > 0$.

The Bayes-optimal fusion rule derived earlier chooses design $1$ when  \eqref{eq:Bayes_optimal} is satisfied.  For notational convenience, let 
\begin{equation*}
S_M(D) =  \sum_{m=1}^M\left[d_m\log\frac{p_m}{q_m} + (1 - d_m)\log\frac{1 - p_m}{1 - q_m} \right].
\end{equation*}
and let the threshold be
\begin{equation*}
\gamma =  \log\frac{L(1,0)}{L(0,1)}\log\frac{\pi_1}{\pi_0}.
\end{equation*}
Then the Bayes fusion rule is $\delta^\ast(D) = 1 \Longleftrightarrow S_M(D) > \gamma$.  Otherwise, $\delta^\ast(D) = 0$.  The following theorem gives a finite-sample information-loss bound and an asymptotic equivalence statement.

\begin{theorem}
Suppose that, conditional on $B$, the local design recommendations $(D_1, \ldots, D_M)$ are independent and satisfy $P(D_m=1\mid B=1)=p_m$ and $P(D_m=1\mid B=0)=q_m$, with $0 < q_m < p_m < 1$, $m = 1, \ldots, M$.  Let $\delta^\ast$ be the Bayes-optimal fusion rule.  Then, the fused distributed design is asymptotically equivalent to the centralized information-optimal design: $P\{\delta^\ast(D) = B\} \to 1$ and $\mathbb{E}\{L(\delta^\ast(D),B)\}\to 0$ exponentially fast as $M \to \infty$.
\label{thrm:loss_asymptotic}
\end{theorem}
\begin{proof}
If the fused decision agrees with the centralized design, then $\delta^\ast(D)=B$, and therefore the realized information loss $L(\delta^\ast(D),B) = 0$.  If the fused decision disagrees with the centralized design, then $\delta^\ast(D)\neq B$, and the loss is either $L(0,1)$ or $L(1,0)$.  Define $L_{\max} = \max\{L(0,1),L(1,0)\}$.  Hence, for every realization of $D$ and $B$,
\begin{equation*}
L(\delta^\ast(D),B) \leq L_{\max}\mathbf{1}\{\delta^\ast(D)\neq B\},
\end{equation*}
where $\mathbf{1}\{\delta^\ast(D) \neq B\}$ is the indicator function of the event $\{\delta^\ast(D) \neq B\}$.  Taking expectations gives 
\begin{equation}
\mathbb{E}\{L(\delta^\ast(D),B)\} \leq L_{\max} P{\delta^\ast(D)\neq B}.  
\end{equation}
Conditioning the error probability on $B$ gives
\begin{equation*}
\begin{split}
P\{\delta^\ast(D)\neq B\} = \pi_1 P\{\delta^\ast(D) = 0 \mid B = 1\}  \\
+ \pi_0 P\{\delta^\ast(D) = 1 \mid B = 0\}.
\end{split}
\end{equation*}
Since $\delta^\ast(D) = 1 \Longleftrightarrow S_M(D) > \gamma$, we have $P\{\delta^\ast(D) = 0 \mid B = 1\} = P_1(S_M\leq \gamma)$ and $P\{\delta^\ast(D) = 1 \mid B = 0\} = P_0(S_M>\gamma)$.  Therefore,
\begin{equation}
P\{\delta^\ast(D) \neq B\} = \pi_1 P_1(S_M\leq \gamma) + \pi_0 P_0(S_M>\gamma),
\label{eq:exact_decomposition}
\end{equation}
where $P_b$ denotes probability conditional on $B=b$.  

We next bound the first probability term $P_1(S_M\leq \gamma)$ in \eqref{eq:exact_decomposition}.  Under $B=1$, the random variables $(D_1, \ldots, D_M)$ are independent Bernoulli random variables with success probabilities $(p_1, \ldots, p_M)$. Thus $P_1(D_m = 1) = p_m$, $P_1(D_m = 0) = 1 - p_m$.  Let
\begin{equation*}
Z_m = D_m\log\frac{p_m}{q_m} + (1 - D_m)\log\frac{1 - p_m}{1 - q_m}.
\end{equation*}
Then, $S_M = \sum_{m = 1}^M Z_m$.  For any $s\geq 0$, $P_1(S_M\leq \gamma) = P_1(e^{-sS_M}\geq e^{-s\gamma})$.  By Markov's inequality, $P_1(e^{-sS_M} \geq e^{-s\gamma}) \leq e^{s\gamma}\mathbb{E}_1(e^{-sS_M})$, where $\mathbb{E}_1$ denotes expectation under the probability measure $P_1$.  Therefore, $P_1(S_M\leq \gamma) \leq e^{s\gamma}\mathbb{E}_1(e^{-sS_M})$.  Since the $D_m$ are conditionally independent under $B=1$, we can write
\begin{equation*}
\mathbb{E}_1(e^{-sS_M}) = \prod_{m = 1}^M \mathbb{E}_1(e^{-sZ_m}),
\end{equation*}
where
\begin{equation*}
e^{-sZ_m} = \left(\frac{p_m}{q_m}\right)^{-sD_m}\left(\frac{1-p_m}{1-q_m}\right)^{-s(1-D_m)}.
\end{equation*}
Hence, we get 
\begin{equation*}
\begin{split}
\mathbb{E}_1(e^{-sZ_m}) &= p_m\left(\frac{p_m}{q_m}\right)^{-s} + (1 - p_m) \left(\frac{1 - p_m}{1 - q_m}\right)^{-s} \\
&= p_m^{1-s}q_m^s + (1-p_m)^{1-s}(1-q_m)^s.
\end{split}
\end{equation*}
Replacing $s$ by $1-s$ gives an equivalent parametrization on $[0,1]$: $\mathbb{E}_1(e^{-(1 - s)Z_m}) = p_m^s q_m^{1 - s} + (1 - p_m)^s(1 - q_m)^{1-s}$.  Thus, for $0\leq s\leq 1$,
\begin{equation}
P_1(S_M\leq \gamma) \leq e^{s\gamma} \prod_{m=1}^M c_m(s),
\label{eq:probability_bound1}
\end{equation}
where $c_m(s) = p_m^s q_m^{1-s} + (1-p_m)^s(1-q_m)^{1-s}$ for each $m$ with $0 \leq s \leq 1$ is the local Chernoff coefficient.  

A parallel argument bounds the second probability term in \eqref{eq:exact_decomposition}. Under $B=0$, $D_m$ is Bernoulli with success probability $q_m$. For $0\leq s\leq 1$, $P_0(S_M > \gamma) = P_0(e^{(1 - s)S_M} > e^{(1 - s)\gamma})$.  By Markov's inequality, $P_0(S_M > \gamma) \leq e^{-(1 - s)\gamma}\mathbb{E}_0(e^{(1 - s)S_M})$, where $\mathbb{E}_0$ is the expectation operator under the probability measure $P_0$.  Using conditional independence, we get
\begin{equation*}
\mathbb{E}_0(e^{(1 - s)S_M}) = \prod_{m=1}^M \mathbb{E}_0(e^{(1 - s)Z_m}),
\end{equation*}
where
\begin{equation*}
\begin{split}
\mathbb{E}_0(e^{(1 - s)Z_m}) &= q_m \left(\frac{p_m}{q_m}\right)^{1 - s} + (1 - q_m) \left(\frac{1 - p_m}{1 - q_m}\right)^{1 - s} \\
&= p_m^{1 - s}q_m^s + (1 - p_m)^{1 - s}(1 - q_m)^s.
 \end{split}
\end{equation*}
Equivalently, with the same Chernoff coefficient after reparameterization,
\begin{equation*}
\mathbb{E}_0(e^{sZ_m}) = p_m^s q_m^{1 - s} + (1 - p_m)^s(1 - q_m)^{1 - s}.
\end{equation*}
Thus,
\begin{equation}
P_0(S_M > \gamma) \leq e^{-(1 - s)\gamma} \prod_{m = 1}^M c_m(s).
\label{eq:probability_bound2}
\end{equation}
Combining the two bounds \eqref{eq:probability_bound1} and \eqref{eq:probability_bound2}, we get 
\begin{equation*}
\begin{split}
P\{\delta^\ast(D)\neq B\} &\leq \pi_1 e^{s\gamma} \prod_{m=1}^M c_m(s) + \pi_0 e^{-(1-s)\gamma}\prod_{m=1}^M c_m(s) \\
&= \left[\pi_1 e^{s\gamma} + \pi_0 e^{-(1-s)\gamma}\right]\prod_{m = 1}^M c_m(s).
\end{split}
\end{equation*}
Since the bound holds for every $s \in [0,1]$, we take the infimum over $s$ to get
\begin{equation*}
P\{\delta^\ast(D)\neq B\} \leq \inf_{0\leq s\leq 1} \left[\pi_1 e^{s\gamma} + \pi_0 e^{-(1-s)\gamma}\right]\prod_{m=1}^M c_m(s).
\end{equation*}
Consequently,
\begin{equation*}
\begin{split}
\mathbb{E}\{L(\delta^\ast(D),B)\} \leq &L_{\max}\inf_{0\leq s\leq 1}\left[\pi_1 e^{s\gamma} + \pi_0 e^{-(1-s)\gamma} \right]
\\
&\times \prod_{m=1}^M c_m(s).
\end{split}
\end{equation*}
Furthermore, in the identically distributed case, $p_m = p$,  $q_m = q$ for all $m$ with $0 < q < p < 1$.   Then, $c_m(s)=c(s)$, where $c(s) = p^s q^{1 - s} + (1 - p)^s(1 - q)^{1- s}$.   Since $p \neq q$, the two Bernoulli distributions are distinct.  Hence, their Chernoff information
\begin{equation*}
C = -\log \inf_{0\leq s \leq 1}c(s) > 0.
\end{equation*}
Therefore, $\prod_{m=1}^M c_m(s) = c(s)^M$.  Choosing $s = s^\ast$ that minimizes $c(s)$ gives $c(s^\ast)^M=e^{-MC}$.   For fixed $\gamma$, $\pi_0$, $\pi_1$, the pre-factor $\pi_1 e^{s^\ast\gamma} + \pi_0 e^{-(1 - s^\ast)\gamma}$ does not grow with $M$.  Hence, we get 
\begin{equation*}
P\{\delta^\ast(D)\neq B\} = O(e^{-MC}).
\end{equation*}
Since $\mathbb{E}\{L(\delta^\ast(D),B)\} \leq L_{\max}P\{\delta^\ast(D)\neq B\}$, we also have 
\begin{equation*}
\mathbb{E}\{L(\delta^\ast(D),B)\} = O(e^{-MC}),
\end{equation*}
where $O(e^{-x})$ denotes a quantity that is bounded in magnitude by a constant times $e^{-x}$ for sufficiently large $x$.  Clearly, $P\{\delta^\ast(D) = B\} \to 1$ and $\mathbb{E}\{L(\delta^\ast(D),B)\}\to 0$ exponentially fast as $M \to \infty$.  This completes the proof of \thrmref{thrm:loss_asymptotic}. 
\end{proof}
\thrmref{thrm:loss_asymptotic} says that if the local design recommendations are even weakly informative about the centralized design decision, then the fusion rule recovers the centralized design with probability tending to one.  The condition $p > q$ has a direct interpretation. It says that a local site is more likely to recommend experiment $1$ when experiment $1$ is globally optimal than when experiment $0$ is globally optimal.  The exponential rate is governed by the Chernoff information between the two Bernoulli laws, $\mathrm{Bernoulli}(p)$ and $\mathrm{Bernoulli}(q)$.  Thus the asymptotic result gives a precise statement: $\mathbb{E}\{L(\delta^\ast(D),B)\} \leq \text{constant}\times e^{-MC}$.  The constant depends on the prior odds and the EIG-loss asymmetry. The exponent depends on how reliably local design recommendations distinguish the globally optimal design.  This is a useful result because it separates two effects: information geometry of the experiment through $L(a,b)$ and reliability of local design recommendations through $C$. 

\section{Two Illustrative Examples}\label{sec:illustrative_models}  
We now illustrate the distributed design-fusion framework through two working models.  These models serve two complementary purposes.  The first, a scalar Gaussian model with two possible experiments, provides a fully explicit case in which the local EIGs, centralized EIG, design recommendations, and information-loss terms can all be computed in closed form.  This example clarifies the structure of the optimal fusion rule and shows that it is generally a reliability-weighted rule rather than a simple majority vote.  The second model treats a nonlinear binary-response experiment.  In that setting, exact EIG is typically not available in closed form, so we use a Fisher-information approximation to obtain an analytically tractable design criterion.  This second example shows that the proposed fusion framework is not tied to Gaussian conjugacy; once local design recommendations and their reliability characteristics are specified, the same information-loss fusion rule applies.  Together, the two models demonstrate both the exact analytical structure of the theory and its extension to more realistic nonlinear experimental settings.

\subsection{Scalar Gaussian model}\label{subsec:scalar_Gaussian}
Let $\theta \sim \mathcal{N}(\mu_0, \tau^2)$.  At site $m$, under design $a \in \{0, 1\}$, suppose $Y_m \mid \theta, a \sim \mathcal{N}(x_{ma}\theta, \sigma_{ma}^2)$, where $x_{ma}$ is the design sensitivity and $\sigma_{ma}^2$ is the measurement-noise variance. Assume conditional independence across sites given $\theta$ and $a$.  The Fisher information about $\theta$ contributed by site $m$ under design $a$ is $J_{ma} = x_{ma}^2/\sigma_{ma}^2$.  If the same system-wide design $a$ is used at all sites, then the total information is
\begin{equation*}
J_a = \sum_{m=1}^M J_{ma} = \sum_{m=1}^M \frac{x_{ma}^2}{\sigma_{ma}^2}.
\label{eq:total_information}
\end{equation*}
Since the model is Gaussian-conjugate, the posterior variance of $\theta$ after observing all sites under design $a$ is
\begin{equation*}
\tau_a^2 = \left(\frac{1}{\tau^2} + J_a\right)^{-1}.
\end{equation*}
The prior entropy is $H(\theta) = 0.5\log(2\pi e\tau^2)$ and the posterior entropy is $H(\theta \mid Y,a) = 0.5\log(2\pi e\tau_a^2)$.  Therefore, the centralized EIG is
\begin{equation*}
\begin{split}
U(a) &= H(\theta)-H(\theta\mid Y,a) = \frac{1}{2}\log\frac{\tau^2}{\tau_a^2} \\
&= \frac{1}{2}\log(1+\tau^2J_a) \\
&= \frac{1}{2}\log\left(1 + \tau^2\sum_{m=1}^M \frac{x_{ma}^2}{\sigma_{ma}^2}\right).
\end{split}
\end{equation*}
The centralized design is $B=1 \Longleftrightarrow U(1) > U(0)$.  Since the logarithm is increasing, $B = 1 \Longleftrightarrow J_1 > J_0$.  That is,
\begin{equation*}
B = 1 \Longleftrightarrow \sum_{m=1}^M \frac{x_{m1}^2}{\sigma_{m1}^2} > \sum_{m=1}^M \frac{x_{m0}^2}{\sigma_{m0}^2}.
\end{equation*}
At site $m$, the local EIG is
\begin{equation*}
U_m(a) = \frac{1}{2} \log\left(1 + \tau^2\frac{x_{ma}^2}{\sigma_{ma}^2}\right).
\end{equation*}
Thus the local design recommendation is $D_m = 1 \Longleftrightarrow U_m(1) > U_m(0)$, or equivalently, 
\begin{equation*}
D_m=1 \Longleftrightarrow \frac{x_{m1}^2}{\sigma_{m1}^2} > \frac{x_{m0}^2}{\sigma_{m0}^2}.
\end{equation*}
The information losses are explicit. If $B=1$, then
\begin{equation*}
L(0,1) = U(1) - U(0) = \frac{1}{2}\log\left(\frac{1+\tau^2J_1}{1+\tau^2J_0}\right).
\end{equation*}
If $B=0$, then
\begin{equation*}
L(1,0) = U(0) - U(1) = \frac{1}{2}\log\left(\frac{1+\tau^2J_0}{1+\tau^2J_1}\right).
\end{equation*}
Therefore, in this Gaussian model, every term in the fusion rule is explicit.  From \thrmref{thrm:Bayes_optimal}, the Bayes-optimal fused design is $\delta^\ast(D)=1$ if and only if
\begin{equation*}
\begin{split}
\log\frac{\pi_1}{\pi_0} + \sum_{m = 1}^M\left[D_m\log\frac{p_m}{q_m} + (1 - D_m)\log\frac{1 - p_m}{1 - q_m}\right] \\
> \log\frac{L(1,0)}{L(0,1)}.
\end{split}
\end{equation*}
The interpretation is simple.  Each site compares the two local signal-to-noise ratios $x_{m1}^2/\sigma_{m1}^2$ and $x_{m0}^2/\sigma_{m0}^2$.  The fusion center then aggregates the local recommendations, but not by majority vote in general. Each recommendation is weighted by its reliability through the log-likelihood ratios $\log(p_m/q_m)$ and $\log (1-p_m/1-q_m)$.  The threshold is shifted by the information loss ratio $L(1,0)/L(0,1)$.  Thus, even in the simplest Gaussian model, the optimal fusion rule is not merely choosing the experiment recommended by most sites. It is a reliability-weighted design-fusion rule with an information-loss threshold.

Suppose the local recommendations are conditionally independent and identically distributed, with $P(D_m=1 \mid B=1) = p$ and $P(D_m=1\mid B=0) = q$, where $0 < q < p < 1$.  Then the fusion statistic becomes
\begin{equation*}
S_M(D) = \sum_{m=1}^M\left[D_m\log\frac{p}{q} + (1 - D_m)\log\frac{1 - p}{1 - q}\right].
\end{equation*}
Let $N_1 = \sum_{m=1}^M D_m$.  Then
\begin{equation*}
S_M(D) = N_1\log\frac{p}{q} + (M - N_1)\log\frac{1 - p}{1 - q}.
\end{equation*}
The optimal fusion rule becomes a threshold rule in $N_1$: $\delta^\ast(D) = 1$ if and only if
\begin{equation*}
N_1 > \frac{\gamma - M\log\left(\frac{1-p}{1-q}\right)}{\log\frac{p}{q} -\log\left(\frac{1-p}{1-q}\right)}.
\end{equation*}
Thus the fused design is a weighted majority rule. It reduces to ordinary majority vote only in special symmetric cases.  The information-loss bound becomes
\begin{equation*}
\mathbb{E}\{L(\delta^\ast(D),B)\} \leq L_{\max}K_\gamma e^{-MC},
\end{equation*}
where
$K_\gamma = \pi_1 e^{s^\ast\gamma} + \pi_0 e^{-(1 - s^\ast)\gamma}$,
and
\begin{equation*}
C = -\log\left[\inf_{0\leq s\leq 1}\left\{p^s q^{1 - s} + (1 - p)^s(1 - q)^{1 - s}\right\}\right].
\end{equation*}
This gives an exact illustration of \thrmref{thrm:Bayes_optimal}.

\subsection{Distributed binary-response experiment}\label{subsec:binary_response}
The Gaussian model is analytically tractable, but many experimental-design problems involve binary outcomes.  A natural next model is a distributed binary-response experiment.  Let $\theta\in\mathbb{R}$ be an unknown scalar parameter with prior $\theta\sim N(\mu_0,\tau^2)$.  At site $m$, under design $a \in \{0, 1\}$, the response is binary: $Y_m \in \{0, 1\}$.  We assume the logistic model given by $P(Y_m = 1 \mid \theta, a) = \ell(\alpha_m + x_{ma}\theta)$, where $\ell(z) = e^z/(1+e^z)$.  Here, $\alpha_m$ is a site-specific baseline effect, and $x_{ma}$ is the design-dependent sensitivity at site $m$.  The conditional density is
\begin{equation*}
p_m(y_m \mid \theta, a) = \ell(\alpha_m + x_{ma}\theta)^{y_m}{1 - \ell(\alpha_m + x_{ma}\theta)}^{1 - y_m}.
\end{equation*}
The exact EIG is
\begin{equation}
U(a) = I(\theta;Y_1,\ldots,Y_M\mid a) = H(\theta) \mathbb{E}_{Y\mid a}\{H(\theta\mid Y,a)\}.
\label{eq:binary_response_EIG}
\end{equation}
Unlike the Gaussian model, \eqref{eq:binary_response_EIG} generally does not reduce to a closed elementary form, because the logistic likelihood is not conjugate to the Gaussian prior.  However, for a scalar parameter and moderately concentrated prior, a standard Fisher-information approximation gives a clean analytic design criterion.  The Fisher information at site $m$ under design $a$ is
\begin{equation*}
J_{ma}(\theta) = x_{ma}^2 \ell(\alpha_m + x_{ma}\theta){1 - \ell(\alpha_m + x_{ma}\theta)}.
\end{equation*}
The prior-averaged Fisher information is
\begin{equation*}
\begin{split}
\bar J_{ma} &= \mathbb{E}_{\theta}\left[x_{ma}^2\ell(\alpha_m + x_{ma}\theta){1 - \ell(\alpha_m + x_{ma}\theta)}\right] \\
&= x_{ma}^2\int_{\mathbb{R}}\ell(\alpha_m + x_{ma}\theta){1 - \ell(\alpha_m + x_{ma}\theta)}\phi(\theta;\mu_0,\tau^2) d\theta.
\end{split}
\end{equation*}
If the same design $a$ is used at every site, then the total prior-averaged Fisher information is $\bar J_a = \sum_{m=1}^M \bar J_{ma}$.  A Laplace approximation to the posterior entropy gives
\begin{equation*}
U(a) \approx \frac{1}{2} \log\left(1+\tau^2\bar J_a\right) = \frac{1}{2}\log\left(1 + \tau^2\sum_{m=1}^M \bar J_{ma}\right).
\end{equation*}
The approximate centralized information-optimal design is $B = 1 \Longleftrightarrow \bar J_1>\bar J_0$, where $\bar J_a = \sum_{m=1}^M \bar J_{ma}$.  At site $m$, the approximate local EIG is $U_m(a) \approx 0.5 \log\left(1 + \tau^2\bar J_{ma}\right)$.  Thus the local design recommendation is $D_m = 1 \Longleftrightarrow \bar J_{m1}>\bar J_{m0}$.  The approximate information loss is
\begin{equation*}
\begin{split}
L(0,1) &\approx \frac{1}{2} \log\left(\frac{1 + \tau^2\bar J_1}{1 + \tau^2\bar J_0}\right), \\
L(1,0) &\approx \frac{1}{2}\log\left(\frac{1 + \tau^2\bar J_0}{1 + \tau^2\bar J_1}\right).
\end{split}
\end{equation*}
\thrmref{thrm:Bayes_optimal} then applies with these EIG-induced losses.

If the prior variance $\tau^2$ is small, then the prior-averaged information can be expanded around $\mu_0$.  Let $\eta_{ma}=\alpha_m+x_{ma}\mu_0$.  Then
\begin{equation*}
\ell(\alpha_m + x_{ma}\theta){1 - \ell(\alpha_m + x_{ma}\theta)} \approx \ell(\eta_{ma}){1 - \ell(\eta_{ma})}.
\end{equation*}
Therefore, $\bar J_{ma} \approx x_{ma}^2\ell(\eta_{ma}){1 - \ell(\eta_{ma})}$.  The local design rule becomes $D_m = 1 \Longleftrightarrow x_{m1}^2 \ell(\eta_{m1}){1 - \ell(\eta_{m1})} > x_{m0}^2\ell(\eta_{m0}){1 - \ell(\eta_{m0})}$, which has a useful interpretation. The experiment is informative when two things happen: (i)  the design sensitivity $|x_{ma}|$ is large; and (ii) the Bernoulli response probability is not close to $0$ or $1$.  The factor $\ell(\eta_{ma}){1-\ell(\eta_{ma})}$ is maximized when $\ell(\eta_{ma}) = 1/2$.  Thus, in binary-response experiments, a design with a large sensitivity may still be poor if it pushes the response probability too close to $0$ or $1$.  The logistic model shows that the most informative design is not necessarily the one with the largest absolute design coefficient. It is the one that balances sensitivity with response variability.

The exact EIG \eqref{eq:binary_response_EIG} is not generally closed form, but the fusion theory remains valid once the local recommendations have been generated.  It only requires the conditional distribution of the local design recommendations given the globally optimal design state.  In the logistic model, suppose $P(D_m = 1 \mid B = 1) = p_m$ and $P(D_m = 1 \mid B = 0) = q_m$, $0 < q_m < p_m <1$.  Then the Bayes-optimal fusion rule is again $\delta^\ast(D) = 1$ if and only if
\begin{equation*}
\begin{split}
\log\frac{\pi_1}{\pi_0} + \sum_{m = 1}^M\left[D_m\log\frac{p_m}{q_m} + (1 - D_m)\log\frac{1 - p_m}{1 - q_m}\right] \\
> \log\frac{L(1,0)}{L(0,1)}.
\end{split}
\end{equation*}
Here $L(1,0)$ and $L(0,1)$ may be defined either using the exact EIG, $U(a)=I(\theta;Y_1,\ldots,Y_M\mid a)$, or using the Fisher-information approximation, $U(a) \approx \frac{1}{2}\log(1+\tau^2\bar J_a)$.  Under the homogeneous reliability assumption, $p_m = p$, $q_m = q$, $0 < q < p < 1$, the information-loss bound is $\mathbb{E}\{L(\delta^\ast(D),B)\} = O(e^{-MC})$, where
\begin{equation*}
C = -\log\inf_{0 \leq s \leq 1}\left[p^s q^{1 - s} + (1 - p)^s(1 - q)^{1 - s}\right].
\end{equation*}
Therefore, the fused distributed design is asymptotically equivalent to the centralized design even in the binary-response model, provided the local recommendations are informative.

The Gaussian model in \secref{subsec:scalar_Gaussian} gives an exact benchmark
\begin{equation*}
U(a) = \frac{1}{2}\log\left(1 + \tau^2\sum_{m = 1}^M\frac{x_{ma}^2}{\sigma_{ma}^2}\right),
\end{equation*}
with explicit local EIG, centralized EIG, information loss, and the fusion threshold.  The logistic model discussed in \secref{subsec:binary_response} gives a nonlinear example: $P(Y_m = 1 \mid \theta, a) = \ell(\alpha_m + x_{ma}\theta)$, where exact EIG is difficult, but Fisher-information EIG gives the tractable approximation
\begin{equation*}
U(a) \approx \frac{1}{2}\log\left(1 + \tau^2\sum_{m = 1}^M\bar J_{ma}\right).
\end{equation*}
This model shows that the proposed theory is not merely a Gaussian artifact. The fusion theorem is distribution-free at the level of local design decisions; only the computation of $U(a)$ and $L(a, b)$ changes from model to model.

\section{Numerical studies}\label{sec:numerical_studies}
This section presents numerical studies designed to illustrate the finite-sample (fixed $M$) and asymptotic (increasing $M$) behaviors of the proposed distributed Bayesian design framework.  We consider the scalar Gaussian and the binary-response experiment models treated in \secref{sec:illustrative_models}.  We compare the centralized oracle\textemdash which has access to all local model information\textemdash with the proposed Bayes fusion rule, majority voting, and random fusion. The experiments are constructed to highlight an important feature of distributed design: the experiment preferred by the majority of local agents need not be the experiment with the largest centralized EIG. In particular, heterogeneous local precisions allow a minority of highly informative sites to dominate the global information criterion. The numerical results presented herein therefore assess whether the Bayes fusion rule can recover near-centralized performance from compressed local messages, and whether simple voting can fail under realistic heterogeneity.

The following computational details apply to the Monte Carlo studies.  For the scalar Gaussian experiment, the fixed-$M$ study used $250{,}000$ training configurations to estimate the Bayes fusion rule and $120{,}000$ independent test configurations to evaluate performance. For the binary-response experiment, the corresponding numbers were $80{,}000$ and $40{,}000$. In the varying-$M$ studies, we used $160{,}000$ training and $80{,}000$ test configurations for the Gaussian model, and $50{,}000$ training and $25{,}000$ test configurations for the binary-response model, for each value of $M$.

The simulation is designed to represent a heterogeneous distributed system. In each Monte Carlo replication, sites are divided into two latent types. Type-0 sites tend to favor experiment 0, whereas type-1 sites tend to favor experiment 1. The important feature of the design is that the two types are not equally informative. In the numerical configuration used here, type-1 sites are less frequent but have larger average precision contributions. Consequently, it is possible for a numerical majority of local sites to recommend experiment 0 while the centralized information gain is larger for experiment 1. This is precisely the regime in which naive voting is expected to fail and in which an information-aware fusion rule is needed.

The local design decision transmitted by site $m$ is 
\begin{equation*}
D_m = \arg\max_{e \in \{0, 1\}} I_m(e).  
\end{equation*}
The fusion center observes only the vector of local decisions $(D_1,\ldots,D_M)$, or equivalently the number $K = \sum_{m = 1}^M D_m$ of local recommendations for experiment 1. It does not observe the local precision values $r_{m, e}$ or the local EIGs $I_m(e)$. The Bayes fusion rule is estimated from Monte Carlo training samples by computing, for each possible value of $K$, $\mathbb E\{I_{\mathrm{cen}}(0) \mid K = k\}$ and $\mathbb{E}\{I_{\mathrm{cen}}(1)\mid K=k\}$, and then selecting the experiment with the larger conditional mean centralized information gain. This is the empirical counterpart of the theoretical Bayes fusion rule derived in \secref{sec:main_results}. The rules that are compared are the centralized oracle, majority voting, and random fusion. The centralized oracle is not an implementable distributed rule; it is included as the benchmark against which information-gain regret is measured.

\begin{table}[htbp!]
\centering
\caption{Finite-sample performance in the scalar Gaussian experiment with $M=25$ local sites.}
\label{tab:gaussian_policy_summary}
\begin{tabular}{|l|r|r|r|r|}
\hline
Rule & \specialcell{Mean\\ EIG} & \specialcell{Mean\\ regret} & \specialcell{90\%\\ regret} & \specialcell{Oracle\\ agreement} \\
\hline
Centralized oracle & 2.326 & 0.000 & 0.000 & 1.000 \\ \hline
Bayes fusion       & 2.316 & 0.010 & 0.020 & 0.876 \\ \hline
Majority vote      & 2.160 & 0.166 & 0.381 & 0.263 \\ \hline
Random fusion      & 2.215 & 0.111 & 0.353 & 0.501 \\ 
\hline
\end{tabular}
\end{table}
\tabref{tab:gaussian_policy_summary} compares the four procedures for the scalar Gaussian experiment with $M=25$ local sites. The centralized oracle provides the unattainable benchmark because it observes all local precision contributions before selecting the global experiment. The Bayes fusion rule observes only the binary local design decisions, yet its mean EIG is $2.316$, compared with $2.326$ for the oracle. Thus the loss from compressing each local design calculation to a single binary recommendation is small when the messages are fused through the information-gain criterion. The mean regret of Bayes fusion is only $0.010$, and its 90th percentile regret is $0.020$, showing that the regret distribution is concentrated near zero.

The same table shows why majority voting is not an appropriate fusion principle for this problem. Majority voting has mean EIG $2.160$ and mean regret $0.166$, an order of magnitude larger than the Bayes-fusion regret. Its oracle agreement probability is only $0.263$, even though it uses the same local binary messages.  The poor performance is a consequence of heterogeneity. In the simulation design, the locally less frequent recommendation is often associated with larger precision contributions. Majority voting therefore answers the wrong question: it identifies the experiment preferred by more sites, not the experiment with larger centralized EIG. Random fusion is included only as a baseline. Its oracle agreement is near $0.5$, as expected, and its mean regret lies between Bayes fusion and majority voting in this particular heterogeneous configuration.

\begin{table}[htbp!]
\centering
\caption{Regret and oracle agreement as the number of local sites increases in the scalar Gaussian experiment.}
\label{tab:gaussian_asymptotic_summary}
\begin{tabular}{|r|r|r|r|r|}
\hline
$M$ & \specialcell{Bayes\\ regret} & \specialcell{Bayes\\ agreement} & \specialcell{Majority\\ regret} & \specialcell{Majority\\ agreement} \\
\hline
 3 & 0.029 & 0.843 & 0.077 & 0.709 \\ \hline
 5 & 0.020 & 0.867 & 0.099 & 0.626 \\ \hline
 7 & 0.024 & 0.835 & 0.116 & 0.551 \\ \hline
11 & 0.019 & 0.843 & 0.134 & 0.456 \\ \hline
15 & 0.015 & 0.852 & 0.147 & 0.384 \\ \hline
21 & 0.011 & 0.873 & 0.161 & 0.299 \\ \hline
31 & 0.008 & 0.883 & 0.171 & 0.217 \\ \hline
41 & 0.006 & 0.905 & 0.178 & 0.162 \\ \hline
61 & 0.003 & 0.933 & 0.184 & 0.096 \\ \hline
81 & 0.002 & 0.954 & 0.186 & 0.060 \\
\hline
\end{tabular}
\end{table}
\tabref{tab:gaussian_asymptotic_summary} studies the effect of increasing the number of local sites. The contrast between Bayes fusion and majority voting becomes sharper as $M$ grows. The Bayes-fusion regret decreases from $0.029$ at $M=3$ to $0.002$ at $M=81$, while its probability of matching the centralized oracle increases from $0.843$ to $0.954$. This is the finite-sample numerical signature of the asymptotic equivalence result: although the fusion center observes only local binary design decisions, the Bayes rule extracts enough information from the decision pattern to recover the centralized information-gain ordering with increasing accuracy.

Majority voting shows the opposite behavior. Its regret increases from $0.077$ at $M=3$ to $0.186$ at $M=81$, and its oracle agreement probability decreases from $0.709$ to $0.060$.  The data-generating mechanism creates a heterogeneous system in which the more numerous class of sites is not the more informative class. As $M$ increases, the majority rule becomes increasingly confident in the recommendation of the larger but less informative group. Therefore, adding more sites amplifies the structural bias of majority voting. The table demonstrates that distributed design cannot be judged by the number of agreeing local agents alone; the relevant quantity is the centralized information gain induced by the fused experimental decision.
\begin{figure}[htbp!]
\centering
\includegraphics[height=1.75in,width=2.5in]{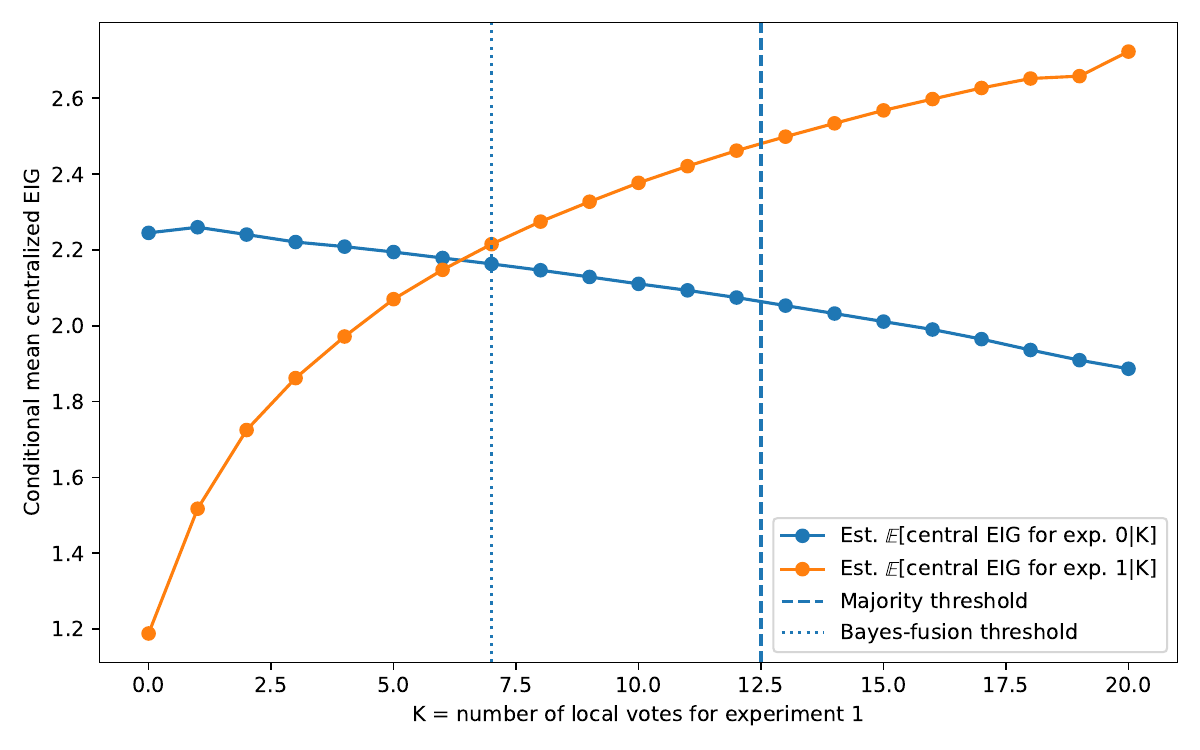}
\caption{Bayes fusion learns an information-gain threshold.}
\label{fig:scalar_info_gain}
\end{figure}
\begin{figure}[htbp!]
\centering
\includegraphics[height=1.75in,width=2.5in]{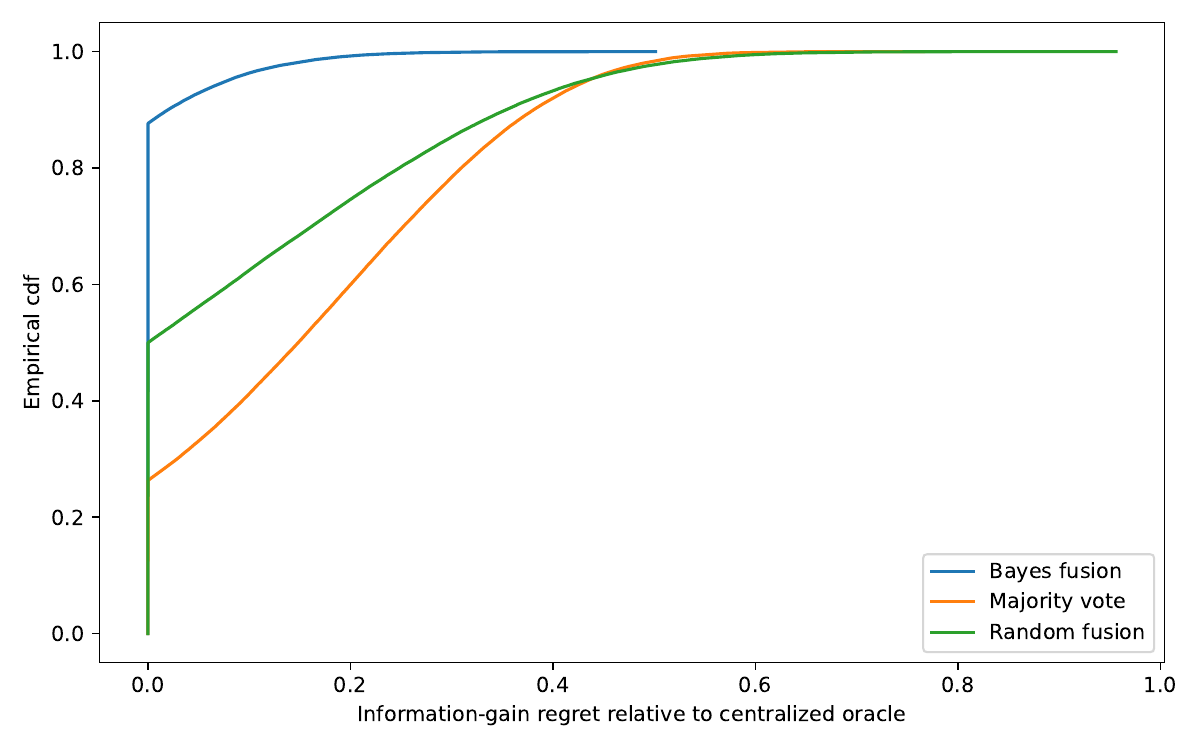}
\caption{Regret distributions.}
\label{fig:scalar_regret}
\end{figure}

The learned Bayes fusion rule is shown in \figref{fig:scalar_info_gain} and has a meaningful interpretation. The majority threshold for $M=25$ would require more than half of the sites to recommend experiment 1 before the fusion center selects experiment 1. The estimated Bayes rule switches to experiment 1 at a smaller value of $K$. This occurs because recommendations for experiment 1 are statistically associated with sites having larger precision contributions. Therefore, a minority of recommendations for experiment 1 can carry enough information to outweigh a majority of weaker recommendations for experiment 0. The fusion rule is not counting votes; it is estimating which experiment has larger centralized information gain conditional on the observed local decision pattern.
The regret distribution shown in \figref{fig:scalar_regret} reinforces this. The empirical distribution of information-gain regret for Bayes fusion is sharply concentrated near zero, indicating that the Bayes rule usually agrees with the centralized oracle or disagrees only in cases where the information loss is small. By contrast, the regret distribution for majority voting has a much heavier upper tail. This means that majority voting not only makes more wrong choices, but also makes wrong choices in cases where the centralized information-gain difference is substantial. In distributed experimental design, this distinction is important: the objective is not to maximize agreement with local decisions, but to preserve global information gain.
\begin{figure}[htbp!]
\centering
\includegraphics[height=1.75in,width=2.5in]{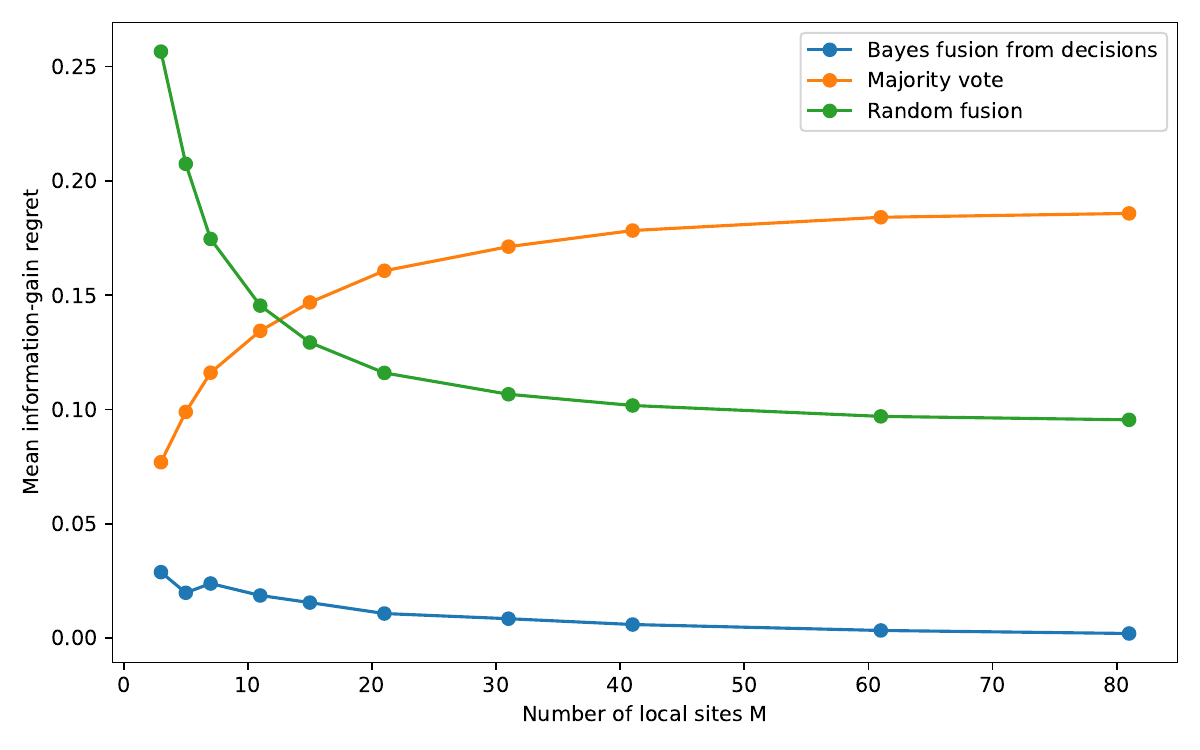}
\caption{Regret relative to centralized oracle.}
\label{fig:scalar_regret_relative}
\end{figure}
\begin{figure}[htbp!]
\centering
\includegraphics[height=1.75in,width=2.5in]{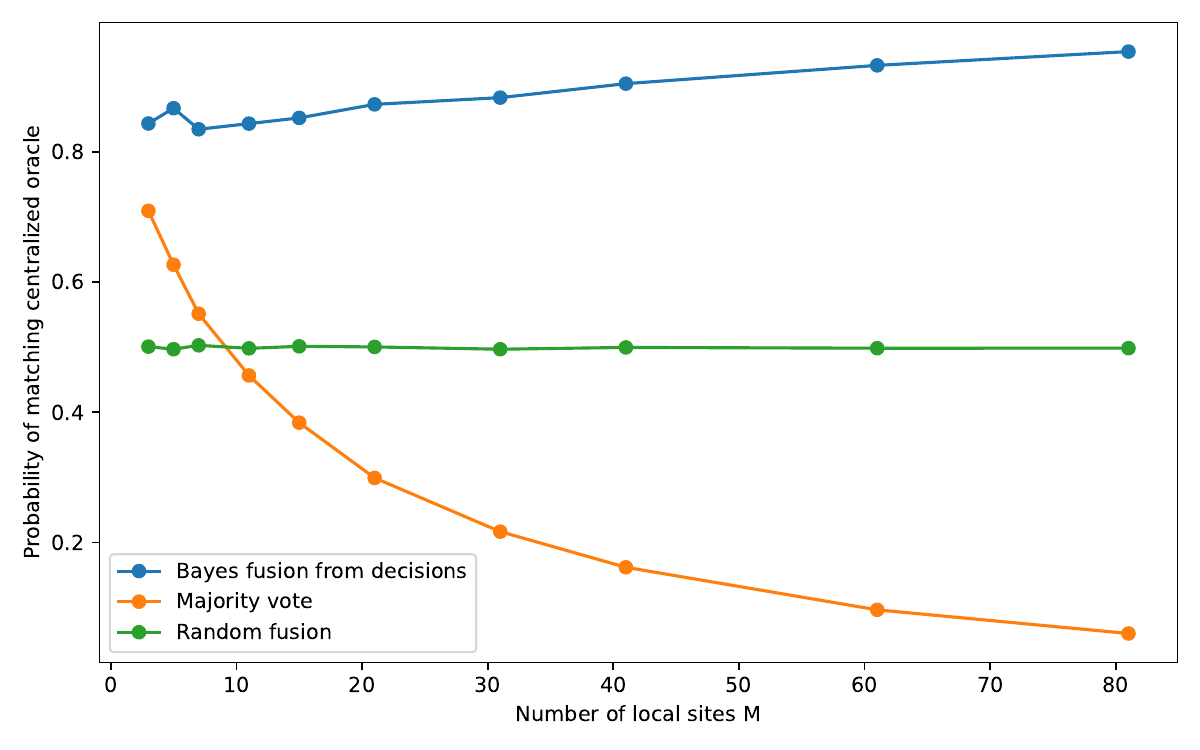}
\caption{Agreement with centralized design decision.}
\label{fig:scalar_agreement}
\end{figure}

The behavior as $M$ increases are shown in \figref{fig:scalar_regret_relative} and \figref{fig:scalar_agreement}, and further illustrates the theoretical point. The mean regret of Bayes fusion decreases toward zero as the number of sites grows, and its agreement with the centralized oracle increases. This provides numerical evidence for asymptotic equivalence between Bayes fusion based on local design decisions and the centralized design. Majority voting exhibits the opposite behavior in this heterogeneous setting. As the number of sites increases, majority voting becomes increasingly committed to the recommendation of the more numerous but less informative class of sites. Its probability of matching the centralized oracle decreases, and its information-gain regret remains substantial. Thus, increasing the number of local agents does not automatically improve a distributed design rule; improvement depends on whether the fusion mechanism is aligned with the centralized information criterion.

Next we consider the distributed binary-response experiment. In this model, the response at site $m$ under experiment $e$ is Bernoulli with success probability $p_{m,e}(\theta)=\{1+\exp[-(\alpha_m+x_{m,e}\theta)]\}^{-1}$. The local information contribution is computed as the prior expectation of the Bernoulli Fisher information, $J_m(e) = \mathbb{E}_\theta[x_{m, e}^2p_{m, e}(\theta)\{1 - p_{m, e}(\theta)\}]$, and the corresponding scalar information utility is taken to be $U_m(e) = 0.5\log\left\{1 + \tau^2J_m(e)\right\}$. The centralized utility replaces $J_m(e)$ by $\sum_{m=1}^M J_m(e)$. This construction preserves the binary-response nature of the data while yielding a tractable information criterion for comparing distributed fusion rules.  

For brevity, we do not display the analogous plots for the binary-response experiment. They show the same qualitative behavior as the scalar Gaussian plots: the Bayes fusion rule closely approximates the centralized oracle, whereas majority voting can be substantially suboptimal when a minority of sites carries greater information.  However, we present two tables quantifying the performances of the binary-response experiment under the finite sample $(M = 25)$ and asymptotic (increasing $M$) regimes. 

\begin{table*}[t!]
\centering
\caption{Monte Carlo performance in the binary-response experiment with $M=25$ local sites.}
\label{tab:binary_response_performance}
\begin{tabular}{|l|r|r|r|r|r|r|}
\hline
\specialcell{Rule} 
& \specialcell{Mean\\ utility} 
& \specialcell{Mean\\ regret} 
& \specialcell{Median\\ regret} 
& \specialcell{90\%\\ regret} 
& \specialcell{Oracle\\ agreement} 
& \specialcell{Select\\ exp. 1} \\
\hline
Centralized oracle 
& 1.8570 & 0.0000 & 0.0000 & 0.0000 & 1.0000 & 0.9282 \\ \hline
Bayes fusion 
& 1.8510 & 0.0060 & 0.0000 & 0.0000 & 0.9421 & 0.9683 \\ \hline
Majority vote 
& 1.5226 & 0.3344 & 0.3446 & 0.6193 & 0.1324 & 0.0606 \\ \hline
Random fusion 
& 1.6648 & 0.1922 & 0.0000 & 0.5734 & 0.5018 & 0.5027 \\
\hline
\end{tabular}
\end{table*}
\tabref{tab:binary_response_performance} reports the Monte Carlo performance of the competing fusion rules in the distributed binary-response experiment with $M=25$ local sites. The centralized oracle achieves a mean utility of $1.8570$ and serves as the benchmark. The proposed Bayes fusion rule attains a mean utility of $1.8510$, with mean regret only $0.0060$ relative to the oracle. It agrees with the centralized oracle in approximately $94.2\%$ of the simulated trials. The median and 90\% regret values are both zero, indicating that when Bayes fusion differs from the oracle, the resulting information loss is usually small.  The contrast with majority voting is substantial. Majority voting has mean utility $1.5226$ and mean regret $0.3344$, much worse than both Bayes fusion and random fusion. Its oracle agreement rate is only $13.2\%$. This poor performance occurs because the binary-response experiment is heterogeneous: the majority of local recommendations may come from less informative sites, while a minority of sites can carry greater information about the parameter. Hence the locally most common recommendation need not coincide with the globally most informative experiment. These results reinforce the main point of the paper: distributed Bayesian design should fuse local recommendations according to EIG, not by simple voting.

\begin{table*}[t!]
\centering
\caption{Performance of fusion rules in the binary-response experiment as the number of local sites $M$ increases.}
\label{tab:binary_response_asymptotic}
\begin{tabular}{|r|r|r|r|r|r|r|r|}
\hline
$M$
& \specialcell{Oracle\\ utility}
& \specialcell{Bayes\\ regret}
& \specialcell{Bayes\\ agree.}
& \specialcell{Majority\\ regret}
& \specialcell{Majority\\ agree.}
& \specialcell{Random\\ regret}
& \specialcell{Random\\ agree.} \\
\hline
 5 & 1.0849 & 0.0354 & 0.8475 & 0.1639 & 0.5461 & 0.2319 & 0.5016 \\ \hline
 9 & 1.3531 & 0.0198 & 0.8711 & 0.2270 & 0.3892 & 0.2071 & 0.5037 \\ \hline
13 & 1.5310 & 0.0134 & 0.8986 & 0.2683 & 0.2919 & 0.1947 & 0.5080 \\ \hline
17 & 1.6634 & 0.0106 & 0.9164 & 0.2985 & 0.2168 & 0.1942 & 0.5000 \\ \hline
21 & 1.7721 & 0.0079 & 0.9286 & 0.3219 & 0.1675 & 0.1935 & 0.5027 \\ \hline
25 & 1.8587 & 0.0059 & 0.9425 & 0.3373 & 0.1304 & 0.1939 & 0.5021 \\ \hline
31 & 1.9667 & 0.0042 & 0.9556 & 0.3535 & 0.0934 & 0.1938 & 0.5005 \\ \hline
41 & 2.1121 & 0.0019 & 0.9758 & 0.3754 & 0.0528 & 0.1971 & 0.5007 \\ \hline
51 & 2.2216 & 0.0010 & 0.9864 & 0.3860 & 0.0284 & 0.2000 & 0.4958 \\
\hline
\end{tabular}
\end{table*}
\tabref{tab:binary_response_asymptotic} reports the behavior of the competing fusion rules in the binary-response experiment as the number of local sites increases. The centralized oracle utility increases with $M$, as expected, because more local sites contribute information about the unknown parameter. The Bayes fusion rule tracks the centralized oracle increasingly closely. Its mean regret decreases from $0.0354$ at $M=5$ to $0.0010$ at $M=51$, while its oracle agreement increases from $0.8475$ to $0.9864$. This pattern is consistent with the asymptotic equivalence result: as the number of local sites grows, the local design messages contain enough information for the Bayes fusion rule to recover the centralized design with high probability.

Majority voting behaves quite differently. Its regret increases with $M$, from $0.1639$ at $M=5$ to $0.3860$ at $M=51$, while its agreement with the oracle decreases from $0.5461$ to $0.0284$. This occurs because the binary-response experiment is heterogeneous: a numerical majority of local recommendations may come from sites with lower information content, while a minority of sites can dominate the centralized information utility. Random fusion remains near $50\%$ agreement with the oracle, as expected, and has a roughly stable regret. Thus the table reinforces the main conclusion of the numerical study: local design recommendations should be fused through the information-gain criterion, not by simple voting.

Thus the numerical evidence supports the main theoretical findings.  Local design decisions can be highly informative compressed summaries, but they must be fused using the information-gain objective. Majority voting is appropriate only when local recommendations are exchangeable with respect to their contribution to centralized information gain. In heterogeneous systems, this exchangeability fails, and vote counts can be misleading. The Bayes fusion rule corrects this by learning how local decision patterns map to centralized EIG.

\section{Conclusion}\label{sec:conclusion}
This paper developed a framework for distributed Bayesian experimental design when local agents transmit only their preferred experiments to a fusion center. The main objective was to understand how such compressed local design decisions should be combined so that the selected global experiment remains close to the centralized expected-information-gain optimum.  We formulated the problem in terms of local EIG, centralized EIG, and information-gain regret. We derived the Bayes fusion rule, which selects the experiment with largest conditional expected centralized information gain given the local design messages. This rule differs from majority voting: the globally informative experiment need not be the one preferred by the largest number of local agents.  We also established information-loss bounds and showed that, under suitable separation and concentration conditions, the distributed fusion rule can become asymptotically equivalent to the centralized design. Numerical studies for scalar Gaussian and binary-response experiments supported the theory. In both settings, Bayes fusion closely approximated the centralized oracle, while majority voting could be substantially suboptimal under heterogeneous local informativeness.  

Future work includes extensions to multiple candidate experiments, richer local messages such as ranked recommendations or quantized information gains, and sequential distributed Bayesian design.  One can also weaken the requirement of common experiment choice at all sites, to enable a form of coded experimentation.  This is especially compelling when experiments are not conditionally independent.  This paper only considers the discrete case where an experiment design is to be chosen from a finite set of candidate designs. The more general continuous design case is open to investigation.  These directions would further clarify the trade-off between communication cost and loss of experimental information.


\balance
\bibliographystyle{IEEEtran}
\bibliography{/Users/kgnagananda/Documents/Work/collaborations/pdx/research/references/research_pdx.bib}

\vfill

\end{document}